\documentclass[11pt]{article}

\usepackage[margin=1in]{geometry}
\usepackage[T1]{fontenc}
\usepackage[utf8]{inputenc}
\usepackage{lmodern}
\usepackage{microtype}
\usepackage{amsmath,amssymb,amsthm}
\usepackage{mathtools}
\usepackage{enumitem}
\usepackage{tikz}
\usepackage[dvipsnames]{xcolor}
\usepackage{hyperref}
\usepackage[nameinlink,noabbrev]{cleveref}
\usepackage{thmtools}
\usepackage{thm-restate}
\usepackage{tcolorbox}

\declaretheorem[name=Theorem,numberwithin=section]{theorem}
\declaretheorem[name=Lemma,sibling=theorem]{lemma}
\declaretheorem[name=Corollary,sibling=theorem]{corollary}

\declaretheorem[name=Definition,sibling=theorem]{definition}

\hypersetup{
  colorlinks=true,
  linkcolor=blue!50!black,
  citecolor=blue!50!black,
  urlcolor=blue!50!black
}

\newcommand{\cV}{\mathcal{V}}
\newcommand{\cH}{\mathcal{H}}
\newcommand{\cS}{\mathcal{S}}

\newcommand{\LP}{\mathrm{LP}}
\newcommand{\R}{\mathbb{R}}
\newcommand{\OPT}{\mathrm{OPT}}
\newcommand{\Prb}{\mathbb{P}}
\newcommand{\Ex}{\mathbb{E}}
\newcommand{\st}{\preceq_{\mathrm{st}}}
\newcommand{\E}{\mathbb{E}}
\DeclareMathOperator{\cost}{cost}

\title{Hitting Axis-Parallel Segments with Weighted Points}
\author{
  Rajiv Raman\thanks{Indraprastha Institute of Information Technology Delhi. Email: \texttt{rajiv@iiitd.ac.in}.}
  \and
  Siddhartha Sarkar\thanks{Indian Institute of Science Bengaluru. Email: \texttt{siddharthas1@iisc.ac.in}.}
  \and
  Jatin Yadav\thanks{Indian Institute of Technology Delhi. Email: \texttt{jatin.yadav@cse.iitd.ac.in}.}
}
\date{}

\begin{document}
\maketitle

\begin{abstract}
We study a geometric hitting-set problem in which the input consists of a set $P$ of
weighted points and a family $\cS=\cH\cup\cV$ of axis-parallel segments in the plane.
The goal is to select a minimum-weight subset of $P$ that hits every segment in $\cS$.
Even restricted geometric hitting-set problems are known to be computationally hard, and
for axis-parallel segments the standard decomposition into horizontal and vertical
sub-instances yields only a simple factor-$2$ approximation.

We present an LP-rounding algorithm that breaks the factor-2 barrier. 
For the weighted problem, we
obtain a randomized $(1+2/e)$-approximation by combining systematic rounding on
horizontal lines with an exact repair step on residual vertical sub-instances. In the
unweighted case, a sharper analysis gives a $(1+1/(e-1))$-approximation.
Finally, we consider the case where one of the sub-instances consists of lines
instead of line segments, a problem considered by Fekete et al. (Geometric Hitting Set for Segments of Few Orientations, Theor. Comp. Sys., 62 (2) 2018),. In this case, we improve their result to
obtain an approximation factor of $1+1/e$ and show that the problem is APX-hard.
We also present algorithms for the generalization to $d$ orientations, as well as
PTASes for bounded-complexity subclasses of the unweighted Hitting Set problem.
\end{abstract}

\section{Introduction}
A \emph{set system} $(P,\mathcal{S})$ 
consists of a set $P$ of elements and a collection
$\mathcal{S}$ of subsets of $P$.
Given a set system, $(P,\mathcal{S})$,
a set $P'\subseteq P$ is a \emph{hitting set} if each set in $\mathcal{S}$
contains an element of $P'$. In the \emph{unweighted} Hitting Set Problem,
the goal is to find a hitting set of minimum cardinality, and in the
\emph{weighted} Hitting Set Problem, we are additionally given a weight
function $w:P\to\R$ and the goal is to select a minimum weight hitting
set $P'$, i.e., a hitting set such that $\sum_{p\in P'} w(p)$ is minimized.
The set systems considered in this paper are finite. We assume that 
$|P|=n$ and $|\mathcal{S}|=m$.

Without restrictions, the Hitting Set Problem is equivalent to the Set Cover Problem on the \emph{dual set system}, where points in $P$ and sets in $\mathcal{S}$ swap roles. For these general systems, $O(\log m)$-approximation algorithms exist with nearly matching lower bounds. However, we can often bypass these limits by exploiting the structure of restricted instances. This paper focuses on \emph{geometric set systems}, where elements and sets are defined by objects in the plane. While Hitting Set and Set Cover remain NP-hard in these settings, they frequently admit better approximation ratios than general systems. Despite this, a unified set of techniques and a complete characterization of these systems remain elusive.

There are two broad frameworks that yield improved approximation
algorithms for Hitting Set and Set Cover. The first is via $\epsilon$-nets and LP-rounding (\cite{HausslerW87,BronnimannG95,EvenRS05, V10, ChanGrant2014})
that yield algorithms with approximation factors depending on the \emph{combinatorial complexity} 
(such as \emph{VC-dimension}, or \emph{shallow-cell complexity})
of the instance. However, these techniques at best yield $O(1)$-approximation algorithms
with large hidden constants even for simple objects such disks or squares.

A second framework is \emph{local search}, that yields PTASes when the geometric set system
is \emph{non-piercing}, but works only in the unweighted setting~\cite{mustafa2010improved, ChanH12, RR18}. 
Informally, a set $\mathcal{S}$ of objects in the plane is non-piercing if the boundary of each object in $\mathcal{S}$ 
is a simple Jordan curve, and for any $A,B\in\mathcal{S}$, $A\setminus B$ and $B\setminus A$ are both
connected sets (see~\cite{RR18} for a precise definition). 
On the other hand,
Chan and Grant~\cite{ChanGrant2014} (see also Har-Peled and Quanrud~\cite{Har-PeledQ15}) showed even with nearly congruent \emph{almost-squares}, the Set Cover and
Hitting Set problems become APX-hard. Recently, Banik et al.~\cite{BanikICALP26} showed that the local search technique
yields a PTAS for the Set Cover problem even when the geometric shapes are not non-piercing, i.e., they \emph{pierce}, but in a bounded manner.
However, no such result is known for the Hitting Set Problem.

Our goal is to develop improved approximation algorithms for structured, weighted geometric set systems. Specifically, we address the Hitting Set problem for axis-parallel line segments (formally defined below).

\begin{tcolorbox}[colback=black!2!white, colframe=black!70!white, title=\textbf{Weighted Geometric Hitting Set for Axis-Parallel Segments}]
\textbf{Input:} A finite set $P$ of points in the plane with a non-negative weight function $w \colon P \to \mathbb{R}_{\ge 0}$, and a family $\cS=\cH\cup\cV$ of axis-parallel segments (where $\cH$ are horizontal and $\cV$ are vertical).

\vspace{1mm}
\noindent\textbf{Goal:} Find a subset $Q\subseteq P$ that minimizes the total weight $\sum_{p \in Q} w(p)$ such that $Q$ hits every segment in $\cS$ (i.e., $Q\cap s\neq\emptyset$ for all $s\in\cS$).
\end{tcolorbox}

Axis-parallel segments generalize several tractable special cases. When all segments in $\mathcal{S}$ are parallel to a single axis, the problem becomes one-dimensional; alternatively, if segments only intersect those of orthogonal orientations, the problem maps to bipartite graphs. In both scenarios, the Hitting Set, Set Cover, and also the Independent Set problems are solvable in polynomial time because their natural LP-relaxations are integral, owing to totally unimodular constraint matrices~\cite{schrijver1998theory}. Crucially, this total unimodularity breaks down in our setting, where segments possess both orientations and parallel segments overlap.

For Hitting Set, Set Cover, and Independent Set, a trivial 2-approximation follows from independently solving the one-dimensional sub-instances for each axis and combining the results. Remarkably, while this remains the best-known approximation factor for Set Cover and Independent Set, our primary contribution is breaking this factor-2 barrier for the weighted Hitting Set problem.

\subsection{Our Contributions}
Our main contribution is a suite of improved approximation algorithms for the Hitting Set problem on axis-parallel line segments and points in the plane, summarized in Table~\ref{table:results}. For line segments in two orientations, we achieve a $(1+2/e) \approx 1.74$-approximation in the weighted setting and a $(1 + 1/(e-1)) \approx 1.58$-approximation in the unweighted setting. When one orientation consists of infinite lines instead of segments, we improve this to a $(1+1/e) \approx 1.36$-approximation. We further generalize our results to segments parallel to one of $d$ directions. These results are obtained via a randomized rounding procedure on the natural linear program, but can be derandomized using the standard method of conditional expectations. For the sake of completeness, we include the derandomization in Appendix~\ref{sec:derandomization}. We also present a PTAS for specific restricted instances.

\begin{table}[h!]
\begin{center}
\begin{tabular}{|l|l|l|l|}
\hline\hline
Problem & Approx. Factor & Complexity & Previous bounds \\
\hline\hline
Wt. axis-parallel   & $1+2/e$ (Th.~\ref{thm:wt})                         &   & $2$ \\
Unwt. axis-parallel & $1 + 1/(e-1)$ (Th.~\ref{thm:unwt})&  &  $2$,  APX-hard~\cite{FeketeEtAl2018}\\
$d$-dir wt.         & $\ln d + \ln\ln d + O(1)$ (Th.~\ref{thm:d-orient}) &  & $d$\\
Hor-seg, vert lines & $1+1/e$ (Thm.~\ref{thm:line})  & APX-hard (Thm.~\ref{thm:apxhard})                    &  $5/3$~\cite{FeketeEtAl2018} \\
$k$-restricted unwt.& PTAS (Thm.~\ref{thm:krestrict})                     &   & NP-hard \\
Lines in 3-directions&   & APX-hard (Thm.~\ref{thm:threeline})&     $4/3$~\cite{DBLP:conf/stoc/DuhF97}, NP-hard~\cite{FeketeEtAl2018}    \\ \hline\hline
\end{tabular}
\caption{The table above summarizes our results.
Note that the previously best-known 2-approximation is considered a folklore result.}
\label{table:results}
\end{center}
\end{table}

\begin{restatable}[Weighted]{theorem}{MainResult}
\label{thm:wt}
Let $\mathcal{S}=\mathcal{H}\cup\mathcal{V}$ be a set of horizontal and vertical line segments
in the plane and let $P$ be a set of points with a weight function $w:P\to\mathbb{R}_{\ge 0}$ 
such that each segment in $\mathcal{S}$ contains at least one point of $P$.
Then, 
there is a randomized $(1+2/e)$-approximation algorithm for the weighted Hitting Set problem
defined by $\mathcal{S}$ and $P$.
\end{restatable}

For the unweighted setting, we obtain a better approximation factor.

\begin{restatable}[Unweighted]{theorem}{UnWeightedResult}
\label{thm:unwt}
Let $\mathcal{S}=\mathcal{H}\cup\mathcal{V}$ be a set of horizontal and vertical line segments
in the plane and let $P$ be a set of points 
such that each segment in $\mathcal{S}$ contains at least one point of $P$.
Then, 
there is a randomized $(1+1/(e-1))$-approximation algorithm for the Hitting Set problem
defined by $\mathcal{S}$ and $P$.
\end{restatable}

We can generalize the result for axis-parallel segments in the setting where each segment is parallel to one of $d$ directions. 
Fekete et al.~\cite{FeketeEtAl2018} studied the setting where each object to be hit is not a single segment,
but is the union of 
$h$ segments (each coming from one of $d$ orientations). The authors provide an $(h\cdot d)$-approximation
for this setting. Applying the result below, it is possible to obtain an improved 
$h\cdot(\log d + \log\log d + O(1))$-approximation.

\begin{restatable}[d-Directions]{theorem}{dorient}
\label{thm:d-orient}
Let $\mathcal{S}$ be a set of segments, each parallel to one of $d$ directions in the plane for some $d \ge 2$, and let $P$ be a set of points with a weight function $w:P\to\mathbb{R}_{\ge 0}$ 
such that each segment in $\mathcal{S}$ contains at least one point of $P$.
Then, there is a randomized algorithm that is a
$
 (\ln d+\ln\ln d+O(1))
$-approximation relative to the natural LP.
\end{restatable}

We show that the natural LP-relaxation of the unweighted problem has an integrality gap of at least $5/4$.

\begin{restatable}[Integrality gap]{theorem}{IntegralityGap}
Let $\mathcal{S}=\mathcal{H}\cup\mathcal{V}$ be a set of horizontal and vertical line segments
in the plane and let $P$ be a set of points 
such that each segment in $\mathcal{S}$ contains at least one point of $P$.
The natural LP-relaxation of the problem has an integrality gap of at least $5/4$.
\end{restatable}

We also consider the setting when the segments in one of the orientations is replaced by lines,
a variant considered by Fekete et al.~\cite{FeketeEtAl2018}, who gave a (5/3)-approximation for the
unweighted case. We obtain a $1+1/e\approx 1.36$-approximation. 

\begin{restatable}[Lines and Segments]{theorem}{LineResult}
\label{thm:line}
Let $\mathcal{S}=\mathcal{H}\cup\mathcal{V}$ be a set of horizontal line segments and vertical lines
in the plane and let $P$ be a set of points with a weight function $w:P\to\mathbb{R}_{\ge 0}$ 
such that each object in $\mathcal{S}$ contains at least one point of $P$.
Then, 
there is a randomized $(1+1/e)$-approximation algorithm for the weighted Hitting Set problem
defined by $\mathcal{S}$ and $P$.
\end{restatable}

This setting with lines in one direction was shown to be NP-hard by Fekete et al.~\cite{FeketeEtAl2018}
if the point set consists of all points in the plane. However, if the set $P$ is finite, we show
that in fact, the problem is APX-hard.

\begin{restatable}[APX-Hardness]{theorem}{Apxhard}
\label{thm:apxhard}
Let $\mathcal{S}=\mathcal{H}\cup\mathcal{V}$ be a set of horizontal line segments and vertical lines
in the plane and $P$ be a finite set of points in the plane such that each object in $\mathcal{S}$ contains at least one point of $P$.
The Hitting Set problem with $\mathcal{S}$ and $P$ is APX-hard.
\end{restatable}

If instead of line segments, we have lines in two dimensions, the Hitting Set problem with a set $P$ of
points is polynomial time solvable via a reduction to Edge Cover in bipartite graphs. 
Fekete~\cite{FeketeEtAl2018} showed that if the set $P$ consists of all points in the plane, then
the problem is NP-hard. We sharpen their result to show that if we are additionally given
a finite set $P$ of points, then the problem becomes APX-hard. The proof can be found
in Appendix~\ref{app:3slc}.

\begin{restatable}[3-Line-APX-hardness]{theorem}{threeline}
\label{thm:threeline}
Let $\cS=\mathcal{X}\cup\mathcal{Y}\cup\mathcal{Z}$ be three sets of lines such that lines in the plane in each
set $\mathcal{X}, \mathcal{Y}$ and $\mathcal{Z}$ are parallel. Let $P$ be a set of points such that each line
contains at least one point. The Hitting Set problem defined by $\cS$ and $P$ is APX-hard.
\end{restatable}

Finally, we consider a restricted setting where each segment intersects segments of the opposite orientation
that lie on at most $k$ distinct parallel lines. We call this the $k$-restricted setting. In this setting, we show
that the problem admits a PTAS.
\begin{restatable}[k-restricted]{theorem}{KRestricted}
\label{thm:krestrict}
Let $\mathcal{S}=\mathcal{H}\cup\mathcal{V}$ be a set of 
horizontal and vertical line segments in the plane 
such that each segment intersects segments of the opposite orientation
lying on at most $k$ parallel lines.
Let $P$ be a set of points 
such that each segment in $\mathcal{S}$ contains at least one point of $P$.
Then, 
there is a PTAS for the Hitting Set problem with $\mathcal{S}$ and $P$.
\end{restatable}

\subsection{Related Work}
\label{sec:related}
\smallskip\noindent
{\bf $\epsilon$-nets:} Geometric Hitting Set and Set Cover approximations heavily depend on shape complexity. The $\epsilon$-net framework\footnote{For a set system $(P,\mathcal{S})$, a set $P'\subseteq P$
is an $\epsilon$-net if it is a hitting set for all sets $S\in\mathcal{S}$ such that $|S|\ge\epsilon|P|$.
} (see~\cite{mustafa2022sampling} for results on $\epsilon$-nets) connects $\epsilon$-net sizes to LP integrality gaps, where bounded VC-dimension yields $O(\log\OPT)$-approximations \cite{BronnimannG95,EvenRS05,HausslerW87}. 
Many geometric set systems have bounded VC-dimension, thus yielding an $O(\log\OPT)$-approximation algorithms.
Since the $\epsilon$-net is obtained by uniform random sampling, the approximation algorithms extend
to the weighted setting as well. While  $\epsilon$-nets of size $O(1/\epsilon)$ 
are known for some objects such as 
\emph{pseudodisks}\footnote{
A set $\mathcal{S}$ of objects in the plane, each of whose boundary is a simple Jordan curve is a set of
pseudodisks if for any $A,B\in\mathcal{S}$, the boundaries of $A$ and $B$ intersect in either 2 or 0 points.
} and halfspaces~\cite{DBLP:conf/compgeom/Matousek90a,PyrgaR08}, translating to $O(1)$-approximation
for Hitting Set, these smaller nets do not yield improved approximation algorithms for the weighted problem. 
If a geometric set system has linear \emph{shallow-cell complexity}, then the \emph{quasi-uniform sampling}
technique~\cite{V10, chan2012weighted} yield improved approximation algorithms for the weighted problems.
However, point-segment systems lack linear shallow-cell complexity.  

\smallskip\noindent{\bf Piercing and non-piercing:}
For the weighted setting, there
is a QPTAS for the Set Cover problem~\cite{mustafa2015quasi} with non-piercing regions in the plane, but no such result is known for the Hitting Set problem.
For the unweighted problem, local search provides PTASes for Hitting Set of pseudodisks,
as shown by Mustafa and Ray~\cite{mustafa2010improved}. The technique also extends to
non-piercing regions~\cite{BasuRoy2018,RR18}, and also yields PTASes for the Set Cover, Independent Set and
Dominating Set problems. 
Recently, Banik et al.~\cite{BanikICALP26} showed that local search also yields PTASes for the Set Cover and Independent Set problems when the objects have a \emph{bounded amount of piercing}~\cite{BanikICALP26}.
Allowing objects to pierce in an unbounded manner however, renders the Hitting Set and Set Cover problems APX-hard, even for nearly congruent shapes like almost-square rectangles~\cite{ChanGrant2014} and nearly equilateral triangles~\cite{Har-PeledQ15}.

\smallskip\noindent{\bf Segments:} 
{\em Axis-Parallel Covering:} Fekete et al.~\cite{FeketeEtAl2018} proved that the Hitting Set problem for axis-parallel segments is APX-hard. For the restricted case where one direction consists of lines, they established NP-hardness and provided a $(5/3)$-approximation. For Set Cover, Kowalska and Pilipczuk~\cite{kowalska2024parameterized} showed APX-hardness for horizontal and vertical segments (even with $\delta$-extensions) and W[1]-hardness for the exact weighted problem. They also developed several FPT algorithms and parameterized approximation schemes parameterized by the solution size $k$.

\smallskip\noindent{\em Independent Set:} For the Independent Set (IS) problem, Marx~\cite{Marx2006iwpec} showed the axis-parallel case is in FPT, whereas arbitrary orientations are W[1]-hard. Parameterized approximation schemes also exist for the axis-parallel setting~\cite{DBLP:conf/esa/CslovjecsekPW24}. Furthermore, Caoduro et al.~\cite{DBLP:journals/jocg/CaoduroCPW23} established a $2-\epsilon$ lower bound on the ratio between clique cover and independence numbers for segment intersection graphs. Consequently, the natural LP-relaxation for IS on axis-parallel segments has an integrality gap of 2, contrasting with our Hitting Set results.

\smallskip\noindent{\em Arbitrary Orientations:} For segments with arbitrary orientations, $\epsilon$-nets have a known size of $O(1/\epsilon \log(1/\epsilon))$ because the VC-dimension is 2. However, Alon~\cite{DBLP:journals/dcg/Alon12} and Balogh et al.~\cite{BaloghSolymosi2018} established an $\epsilon$-net lower bound of $\Omega(1/\epsilon\log^{1/3}(1/\epsilon))$, limiting the best possible $\epsilon$-net-based approximation to $O(\log^{1/3-o(1)}\OPT)$. Finally, for the IS problem on arbitrary segments, there is a weighted QPTAS~\cite{AdamaszekHarPeledWiese2019} and an unweighted $O(n^{\epsilon})$-approximation~\cite{DBLP:journals/cpc/FoxP14}.

\smallskip\noindent{\bf Rectangles:}
Axis-parallel segments can also be seen as a special case of axis-parallel rectangles.
For the Hitting Set problem with axis-parallel rectangles, Aronov, Ezra and Sharir~\cite{DBLP:journals/siamcomp/AronovES10}
obtained an $O(\log\log\OPT)$-approximation by showing that this set system admits 
an $\epsilon$-net of size $O(1/\epsilon\log\log 1/\epsilon)$, and this is the
best approximation factor known. Pach and Tardos~\cite{pach2011tight}
showed that in fact, this bound is tight, and therefore, the natural LP-relaxation for
the Hitting Set problem has an integrality gap of $\Omega\left(\log\log\OPT\right)$.

\paragraph{Organization.}
\Cref{sec:preliminaries} presents the main algorithm and analysis. \Cref{sec:gap} gives an explicit integrality-gap example. The extension to $d$-orientations is in \Cref{sec:orientation}.
\Cref{sec:ptas} records a separator-based PTAS for bounded-complexity subclasses. \Cref{sec:apxhardness} gives the hardness results, and we
conclude in \Cref{sec:conclusion}.

\section{LP-rounding Algorithm}
\label{sec:preliminaries}
Let $P$ be a set of $n$ points in the plane with nonnegative weight function
$w \colon P \to \mathbb{R}_{\ge 0}$. Let $\cS=\cH\cup\cV$ be a set of
$m$ axis-parallel segments, where $\cH$ is the set of horizontal segments and $\cV$ is the set
of vertical segments. 
For a segment $s$, let $P\cap s$ denote the candidate points lying on $s$. We always
assume $P\cap s\neq\emptyset$ for every $s\in\cS$. When we refer to a horizontal or
vertical line, we mean a line that contains at least one point of $P$.
The natural LP relaxation for the Hitting Set Problem is
\[
\begin{aligned}
\min \quad & \sum_{p\in P} w(p)\,x_p \\
\text{s.t.}\quad & \sum_{p\in P\cap s} x_p \ge 1 \qquad \forall s\in \cS,\\
& 0\le x_p\le 1 \qquad \forall p\in P.
\end{aligned}
\]
Fix an optimal LP solution $x$, and let
\[
\LP := \sum_{p\in P} w(p)\,x_p.
\]
Since this is a relaxation of the hitting-set problem, $\LP\le \OPT$.

We use one elementary fact repeatedly: if all candidate points and all segments lie on a
single line, then the induced stabbing problem is an interval stabbing instance, whose
constraint matrix has the consecutive-ones property and is therefore totally unimodular~\cite{schrijver1998theory}.
Consequently, the one-dimensional problem can be solved exactly in polynomial time.

\medskip\noindent{\bf Algorithm:}
The algorithm has two phases. In the first phase, we use randomized rounding to obtain
a hitting set for the horizontal segments. In doing so, we ensure that each point $p$
is in the hitting set produced with probability $x_p$, its LP-value. 
This solution partitions the vertical segments into \emph{blocks}. Each block $B$
is a one-dimensional problem, and the LP solution restricted to $B$ is feasible
for the intervals contained in $B$. Therefore, we can obtain a solution of value at most $\sum_{p\in B}w(p)x_p$.
While the algorithm remains the same for the weighted and unweighted problems, the analysis for the unweighted problem can be somewhat sharpened.

We now describe Phase~I, where we obtain a hitting set for the horizontal segments.
Fix a horizontal line $h$ and let $p_1,\ldots, p_k$ be the points on $h$ from left
to right. We define the cumulative sums $a_i$, $i=0,\ldots, k$ as follows:
\[
a_0:=0,
\qquad
a_i:=\sum_{j=1}^{i} x_{p_j}
\quad (i=1,\dots,k),
\]
For each $i$ let $I_i := [a_{i-1},a_i)$. 
Choose a random shift $U_h$ uniformly from $[0,1]$. We select point $p_i$ in Phase~I if
and only if $I_i$ contains an integer of the shifted lattice $U_h+\mathbb{Z}$. We perform
this rounding independently for every horizontal line.

We claim that each point $p$ is selected in Phase~I with probability $x_p$, and
that the points selected in Phase~I constitute a hitting set for the horizontal segments.

\begin{lemma}
\label{lem:unif}
For each point $p\in P$, the probability that $p$ is selected in Phase~I is exactly
$x_p$.
\end{lemma}
\begin{proof}
If $p=p_i$, then the interval $I_i$ has length $x_{p_i}\le 1$. As the shift $U_h$ varies
uniformly over $[0,1)$, the set of shifts for which $I_i$ contains a point of
$U_h+\mathbb{Z}$ has measure exactly $|I_i|=x_{p_i}$. Hence
$\Prb[p_i \text{ is selected in Phase~I}]=x_{p_i}$.
\end{proof}

\begin{lemma}
\label{lem:hor}
Every horizontal segment is hit by Phase~I with probability $1$.
\end{lemma}

\begin{proof}
Let $s\in \cH$ be a horizontal segment on line $h$, and suppose it contains the
consecutive points $p_i,p_{i+1},\dots,p_j$. The union of the corresponding intervals is
$[a_{i-1},a_j)$, whose length is
\[
a_j-a_{i-1} = \sum_{t=i}^{j} x_{p_t} \ge 1
\]
by the LP constraint for $s$. Any interval of length at least $1$ contains a point of
$U_h+\mathbb{Z}$ for every shift $U_h\in[0,1)$. Therefore at least one point of
$s\cap P$ is selected. 
\end{proof}

In Phase~II, we solve the residual problem obtained on removing the points selected in Phase~I and the segments in $\mathcal{S}$ hit by these points. This residual problem consists of a disjoint collection of vertical segments, each of which is called a \emph{block}. Each block consists of a maximal consecutive run of unselected points on a vertical line and the segments that lie entirely in the block. Each block can be solved optimally either
via dynamic programming, or by solving a natural LP-relaxation for the block.
We now analyze the cost of the solution returned by the algorithm. First, we observe that the expected cost of the points selected in Phase~I is exactly equal to the LP cost.

\begin{lemma}
\label{lem:ph1}
$\Ex[\cost(\text{Phase~I})] = \LP$
\end{lemma}
\begin{proof}
$\Ex[\cost(\text{Phase~I})] = \sum_{p\in P} w(p)\Prb[p \text{ selected in Phase~I}] =
\sum_{p\in P} w(p)x_p = \LP$, where the second-last equality follows from Lemma~\ref{lem:unif}.
\end{proof}

We analyze Phase~II separately for the weighted case, the unweighted case, and the case where vertical segments are replaced by lines.

\subsection{Weighted Hitting Set}
\label{sec:weighted}

Consider a block $B$ in Phase~II. Let $x(B)$ denote the sum of the LP values of all the points contained in the block. If $x(B) < 1$, then observe
that no segment is completely contained in $B$. We can therefore, safely ignore such blocks and only deal with blocks with $x(B)\ge 1$. Let the weight of a block $B$ be $w(B)=\sum_{i\in B} w(i)x_i$. For a point $p\in P$, let $A_p$ be the event that $p$ is unselected in Phase~I and the block containing $p$ has total LP mass at least $1$. Summing the bound above over all blocks and charging each block to its LP mass gives
\[
\Ex[\cost(\text{Phase~2})]
\le
\sum_{p\in P} w(p)x_p\,\Prb[A_p].
\]
We now obtain an upper bound on $\Prb[A_p]$.
Fix a vertical line $\ell$ with ordered points $1,2,\dots,m$ from bottom to top, and fix a
point $p\in\{1,\dots,m\}$. Define random variables $D_p$ and $U_p$ as follows.
\begin{itemize}
  \item $D_p$ is the total LP mass of the maximal consecutive run of unselected points
  ending at $p$. If $p$ itself is selected in Phase~I, then $D_p=0$.
  \item $U_p$ is the total LP mass of the maximal consecutive run of unselected points
  immediately above $p$.
\end{itemize}
If $p$ is unselected, then $D_p+U_p$ is exactly the total LP mass of the block containing
$p$. Therefore, $\Prb[A_p]\le \Prb[D_p+U_p\ge 1]$. We now show that the probability that $D_p$ or $U_p$ has mass at least $t$ decreases exponentially.

\begin{lemma}
\label{lem:LpRp}
For every $t\ge 0$, $\Prb[D_p\ge t]\le e^{-t}$ and $\Prb[U_p\ge t]\le e^{-t}$.
Moreover, $D_p$ and $U_p$ are independent.
\end{lemma}

\begin{proof}
We prove the claim for $D_p$; the proof for $U_p$ is identical.
Since the claim is trivially true for $t = 0$, we consider $t > 0$. If no suffix of points ending at $p$ has total LP mass at least $t$, then
$\Prb[D_p\ge t]=0$. Otherwise, let $j\le p$ be the largest index such that
$
\sum_{i=j}^{p} x_i \ge t.
$
If $D_p\ge t$, then all points $j,j+1,\dots,p$ must be unselected in Phase~I. Since Phase~I
decisions on a vertical line are independent,
\[
\Prb[D_p\ge t]
\le
\prod_{i=j}^{p}(1-x_i)
\le
\exp\!\left(-\sum_{i=j}^{p}x_i\right)
\le
e^{-t}.
\]

Finally, $D_p$ depends only on the decisions for points at or below $p$, while $U_p$
depends only on the decisions strictly above $p$. These sets of decisions are independent,
so $D_p$ and $U_p$ are independent.
\end{proof}

We require the following technical lemma on stochastic dominance, whose proof is in
Appendix~\ref{app:stochastic-dominance}.

\begin{restatable}{lemma}{StochasticSumDominance}
\label{lem:stochastic_sum_dominance}
Let $X_1,X_2,Y_1,Y_2$ be random variables such that $X_1$ is independent of
$X_2$ and $Y_1$ is independent of $Y_2$. If $X_1\st Y_1$ and $X_2\st Y_2$, then
$X_1+X_2 \st Y_1+Y_2$.
\end{restatable}

Let $E_1,E_2\sim \mathrm{Exp}(1)$ be independent exponential random variables. The tail
bound above implies $D_p\st E_1$ and $U_p\st E_2$. By
\Cref{lem:stochastic_sum_dominance}, we have
$D_p+U_p \st E_1+E_2$. Therefore,
\[
\Prb[D_p+U_p\ge 1]
\le
\Prb[E_1+E_2\ge 1].
\]
Since $E_1+E_2$ has the $\Gamma(2,1)$ distribution,
\[
\Prb[E_1+E_2\ge 1] = e^{-1}(1+1)=\frac{2}{e}.
\]
Hence $\Prb[A_p]\le 2/e$ for every point $p$, and therefore
\[
\Ex[\cost(\text{Phase~2})]
\le
\frac{2}{e}\sum_{p\in P} w(p)x_p
=
\frac{2}{e}\LP.
\]

\MainResult*
\begin{proof}
Combining the bound on Phase~I from Lemma~\ref{lem:ph1}, and the bound on $A_p$ above,
we obtain
\[
\Ex[\cost(\text{output})]
\le
\LP + \frac{2}{e}\LP
=
\left(1+\frac{2}{e}\right)\LP
\le
\left(1+\frac{2}{e}\right)\OPT.
\qedhere
\]
\end{proof}

\subsection{Unweighted Hitting Set}
\label{sec:unweighted}
Fix a vertical line $\ell$ containing at least one vertical segment in $\cV$. Let the points on this line be $q_1,\ldots, q_m$ ordered bottom to top.
Let $t$ denote the highest point such that 
the sum of the LP values of the points above $t$ is at least  $1$. That is,
$t = \max\left\{j\in\{1,\dots,m\} : \sum_{i=j}^{m} x_i \ge 1\right\}$.
Since the line contains at least one segment, the LP value on $\ell$ is $\geq 1$. Therefore, we can assume that $t$ is defined.
In the unweighted setting, we can improve the analysis by observing that the value
of an optimal solution for a block $B$ is at most $\lfloor\sum_{p\in B}x_p\rfloor$.

For ease of analysis, we 
introduce a sentinel point $0$ immediately below point $1$, and declare it always
selected. For every $i\in\{0,1,\dots,m\}$, let $B(i)$ denote the maximal block of
consecutive unselected points immediately above $i$, and define
\[
Y_i := \lfloor x(B(i)) \rfloor,
\]
where $x(B(i)):=\sum_{p\in B(i)}x_p$. Every block on the line is either $B(0)$ or a block
$B(i)$ following a selected point $i\ge 1$. Thus the Phase~II cost on this line is at most
\[
Y_0 + \sum_{i=1}^{m} \mathbf{1}[\text{$i$ selected}]\,Y_i,
\]
where $\mathbf{1}[\text{$i$ selected}]$ is the indicator random variable that is 1 if
$i$ is selected, and $0$ otherwise. 
Moreover, if $i\ge t$, then $\sum_{j=i+1}^{m}x_j<1$, so $Y_i=0$. Therefore only
$Y_0,Y_1,\dots,Y_{t-1}$ matter.

\begin{lemma}
\label{lem:yi}
For every $i\in\{0,1,\dots,t-1\}$,
\[
\Ex[Y_i \mid \text{$i$ is selected}] \le \frac{1}{e-1},
\]
where for $i=0$ the conditioning is vacuous.
\end{lemma}

\begin{proof}
Fix $i\in\{0,1,\dots,t-1\}$. For every integer $k\ge 1$,
\[
\Prb\bigl[x(B(i))\ge k \mid \text{$i$ is selected}\bigr] \le e^{-k}.
\]
Indeed, if the above probability is non-zero, let $r>i$ be the smallest index such that$\sum_{j=i+1}^{r}x_j \ge k$. Then all points $i+1,\dots,r$ must not be selected in Phase~I. Since the rounding decisions
on a vertical line are independent and the event that $i$ is selected depends only on
point $i$, conditioning on the selection of $i$ does not affect the points above it. Thus
\[
\Prb\bigl[x(B(i))\ge k \mid \text{$i$ is selected}\bigr]
\le
\prod_{j=i+1}^{r}(1-x_j)
\le
\exp\!\left(-\sum_{j=i+1}^{r}x_j\right)
\le
e^{-k}.
\]
Using the tail-sum formula\footnote{If $X$ is a non-negative integer-valued random variable ($X \in \{0, 1, 2, \dots\}$), the expected value is: $\Ex[X] = \sum_{n=1}^{\infty} \Prb(X \ge n)$},
\[
\begin{aligned}
\Ex[Y_i \mid \text{$i$ is selected}]
&=
\sum_{k\ge 1}\Prb\bigl[Y_i\ge k \mid \text{$i$ is selected}\bigr] \\
&\le
\sum_{k\ge 1}\Prb\bigl[x(B(i))\ge k \mid \text{$i$ is selected}\bigr] \\
&\le
\sum_{k\ge 1}e^{-k}
=
\frac{1}{e-1}.
\end{aligned}
\]
\end{proof}

\begin{lemma}
For every vertical line,
\[
\Ex[\cost(\text{Phase~II on this line})]
\le
\frac{1}{e-1}\sum_{i=1}^{m}x_i.
\]
\end{lemma}

\begin{proof}
Using the block decomposition above,
\[
\Ex[\cost(\text{Phase~II on this line})]
\le
\Ex[Y_0] + \sum_{i=1}^{t-1}\Prb[\text{$i$ selected}]\,\Ex[Y_i\mid \text{$i$ selected}].
\]
By Lemma~\ref{lem:yi}, every conditional expectation is at most $1/(e-1)$, and point $i$
is selected with probability $x_i$. Hence
\[
\Ex[\cost(\text{Phase~II on this line})]
\le
\frac{1}{e-1} + \sum_{i=1}^{t-1}\frac{x_i}{e-1}
=
\frac{1+\sum_{i=1}^{t-1}x_i}{e-1}.
\]
By the definition of $t$, we have $\sum_{i=t}^{m}x_i\ge 1$, so
\[
1+\sum_{i=1}^{t-1}x_i \le \sum_{i=1}^{m}x_i.
\]
Substituting this inequality completes the proof.
\end{proof}

\UnWeightedResult*
\begin{proof}
Phase~I still has expected cost exactly $\LP=\sum_{p\in P}x_p$. Summing the previous line
bound over all vertical lines gives
\[
\Ex[\cost(\text{Phase~2})] \le \frac{1}{e-1}\LP.
\]
Therefore
\[
\Ex[\cost(\text{output})]
\le
\LP + \frac{1}{e-1}\LP
=
\left(1+\frac{1}{e-1}\right)\LP
\le
\left(1+\frac{1}{e-1}\right)\OPT.
\qedhere
\]
\end{proof}

The analysis above is asymptotically tight; see
Appendix~\ref{app:unweighted-tightness}.

\subsection{Horizontal Segments and Vertical Lines}
\label{sec:lines}
In this section, we show that the approximation factor can be improved when in one of the directions, we have
lines instead of line segments.

    \LineResult*
\begin{proof}
    Each vertical line $\ell$ that is not hit by the Phase~I points contributes $\min_{p \in \ell} w(p) \leq \sum_{p \in \ell} w(p) x_p$ to the Phase~II cost, and each vertical line that is hit by the Phase~I points contributes $0$. Thus,
    \[
    \Ex[\cost(\text{Phase~II})] \le \sum_{\ell \in \cV} \Prb[\ell \text{ is not hit by Phase~I}] \sum_{p \in \ell} w(p) x_p. 
    \]
    Now, since $\sum_{p \in \ell} x_p \geq 1$, we have that $\Prb[\ell \text{ is not hit by Phase~I}] \leq \frac{1}{e}$, and hence the total expected cost of Phase~II is upper bounded by $\frac{\LP}{e}$. Thus, from~\Cref{lem:ph1}, we get that the total expected cost is at most $(1 + 1/e)\LP$. \qedhere
\end{proof}

\section{An Explicit Integrality-Gap Construction}
\label{sec:gap}
In this section, we present an instance on which the integrality gap of the standard LP for the unweighted problem is $1.25$.

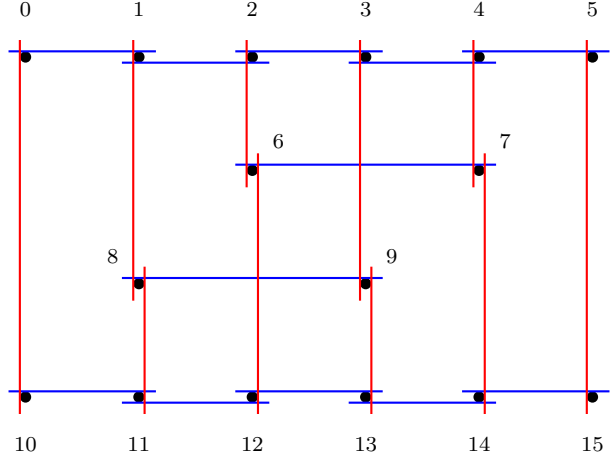
\begin{figure}[h!]
\centering
\begin{tikzpicture}[scale=1.5,
  point/.style={circle, fill=black, inner sep=1.5pt},
  labeltext/.style={font=\scriptsize, fill=white, inner sep=1pt},
  hsegment/.style={thick, blue},
  vsegment/.style={thick, red}]

\foreach \x/\n in {0/0, 1/1, 2/2, 3/3, 4/4, 5/5} {
    \node[point] (\n) at (\x, 3) {};
    \node[labeltext] at (\x, 3.42) {\n};
}
\foreach \x/\n in {0/10, 1/11, 2/12, 3/13, 4/14, 5/15} {
    \node[point] (\n) at (\x, 0) {};
    \node[labeltext] at (\x, -0.42) {\n};
}

\node[point] (6) at (2, 2) {};
\node[labeltext] at (2.23, 2.26) {6};
\node[point] (7) at (4, 2) {};
\node[labeltext] at (4.23, 2.26) {7};
\node[point] (8) at (1, 1) {};
\node[labeltext] at (0.77, 1.24) {8};
\node[point] (9) at (3, 1) {};
\node[labeltext] at (3.23, 1.24) {9};

\draw[hsegment] (-0.15, 3.05) -- (1.15, 3.05); 
\draw[hsegment] (0.85, 2.95) -- (2.15, 2.95);  
\draw[hsegment] (1.85, 3.05) -- (3.15, 3.05);  
\draw[hsegment] (2.85, 2.95) -- (4.15, 2.95);  
\draw[hsegment] (3.85, 3.05) -- (5.15, 3.05);  

\draw[hsegment] (-0.15, 0.05) -- (1.15, 0.05); 
\draw[hsegment] (0.85, -0.05) -- (2.15, -0.05); 
\draw[hsegment] (1.85, 0.05) -- (3.15, 0.05);  
\draw[hsegment] (2.85, -0.05) -- (4.15, -0.05); 
\draw[hsegment] (3.85, 0.05) -- (5.15, 0.05);  

\draw[hsegment] (1.85, 2.05) -- (4.15, 2.05);  
\draw[hsegment] (0.85, 1.05) -- (3.15, 1.05);  

\draw[vsegment] (-0.05, 3.15) -- (-0.05, -0.15); 
\draw[vsegment] (0.95, 3.15) -- (0.95, 0.85);    
\draw[vsegment] (1.05, 1.15) -- (1.05, -0.15);   
\draw[vsegment] (1.95, 3.15) -- (1.95, 1.85);    
\draw[vsegment] (2.05, 2.15) -- (2.05, -0.15);   
\draw[vsegment] (2.95, 3.15) -- (2.95, 0.85);    
\draw[vsegment] (3.05, 1.15) -- (3.05, -0.15);   
\draw[vsegment] (3.95, 3.15) -- (3.95, 1.85);    
\draw[vsegment] (4.05, 2.15) -- (4.05, -0.15);   
\draw[vsegment] (4.95, 3.15) -- (4.95, -0.15);   

\end{tikzpicture}
\caption{An integrality-gap instance for unweighted axis-parallel segment hitting.  Segments are drawn with slight offsets for readability. The blue segments denote the horizontal segments $\cH$, and the red segments denote the vertical segments $\cV$.}
\label{fig:gap}
\end{figure}

Let $P=\{0,1,\dots,15\}$ be the candidate points shown in \Cref{fig:gap}. Segments in the instance are formed by pairs of points that are on the same horizontal or vertical level, and have no other point on the segment connecting them. Thus, along the top chain we include every segment corresponding to consecutive pairs of points from $\{0,1,2,3,4,5\}$, along the bottom chain every segment corresponding to consecutive pairs of points from $\{10,11,12,13,14,15\}$, and so on.

\begin{theorem}
\label{thm:gap}
 The instance of \Cref{fig:gap} has integrality gap exactly $5/4$ for the natural LP
relaxation.
\end{theorem}

\begin{proof}
Our hitting set instance is precisely an instance of vertex cover where the vertices are the points and for each segment, there is an edge between the vertices corresponding to the two points on the segment. Let $G$ denote this graph. Set $x_p:=1/2$ for every point $p\in P$. Every segment in the construction contains at
least two marked points, so this assignment is feasible. Hence the LP value is at most
\[
\sum_{p\in P} x_p = 16\cdot \frac12 = 8.
\]

In fact this bound is tight. The eight pairwise vertex-disjoint segments
\[
\{0,1\},\ \{2,3\},\ \{4,5\},\ \{6,7\},\ \{8,9\},\ \{10,11\},\ \{12,13\},\ \{14,15\}
\]
already force $\sum_{p\in P} x_p \ge 8$, so $\LP=8$. It remains to compute the optimum integral value.  Since this is a vertex-cover instance,
\[
\OPT= |V(G)|-\alpha(G)=16-\alpha(G),
\]
where $\alpha(G)$ is the maximum independent-set size.  The set $\{0,2,4,8,12,14\}$ is independent, so $\alpha(G)\ge6$. For the reverse inequality, consider the 6-cycle $C = \{2,3,9,13,12,6\}$ and the two 5-cycles $L=\{0,1,8,11,10\}$ and $R=\{4,5,15,14,7\}$. Every independent set contains at most 3 vertices from $C$ and at most 2 vertices from each of $L$ and $R$.  If an independent set $I$ has $|I\cap C|\le2$, then immediately $|I|\le 2+2+2=6$. 

Suppose instead that $|I\cap C|=3$.  Then $I\cap C$ must be one of the two alternating triples $\{2,9,12\}$ or $\{3,6,13\}$.
If $I$ contains $\{2,9,12\}$, then the edges $(1,2)$, $(8,9)$, and $(11,12)$ exclude $1,8,11$ from the left 5-cycle.  Only $0$ and $10$ remain available in $L$, and they are adjacent, so $|I\cap L|\le1$.  Since $|I\cap R|\le2$, we get $|I|\le3+1+2=6$. The case of the other alternating triple ($\{3, 6, 13\}$) is symmetric.

Therefore $\alpha(G)=6$, and hence $\OPT=16-6=10$.  The integrality gap is
\[
\frac{\OPT}{\LP}=\frac{10}{8}=\frac54.
\]

\end{proof}

\section{Generalization to \texorpdfstring{$d$}{d} orientations (weighted case)}
\label{sec:orientation}
Suppose segments are allowed to have $d$ distinct orientations
$\theta_1,\theta_2,\dots,\theta_d$, and let
$S^{(r)}$ denote the set of segments of orientation $\theta_r$.
The LP relaxation is unchanged:
\[
\min \sum_{p\in P} w(p) x_p
\qquad
\text{s.t. }\sum_{p\in P\cap s} x_p\ge 1\ \forall s\in \bigcup_{r=1}^d S^{(r)},\quad 0\le x_p\le 1.
\]
Let $x$ be an optimal LP solution with value $\LP$.

\noindent\textbf{Algorithm with a parameter $k$.}
Fix $k\in\{1,2,\dots,d\}$ and call $\theta_1,\dots,\theta_k$ the primary orientations.  For every line of each primary orientation, independently run the same random-start systematic rounding used above, after rotating coordinates along that line.  The output of the primary step is the union of all points selected by these $k$ rounds.  Therefore its expected cost is at most
\[
\sum_{p\in P} w(p)\Prb[p\text{ is selected in at least one primary round}]
\le \sum_{p\in P} w(p)\sum_{t=1}^k x_p
= k \cdot \LP.
\]
All segments in the primary orientations are hit with probability one.

For every remaining orientation $r>k$, we then repair the residual one-dimensional instances on lines of orientation $\theta_r$.  These repairs may be performed independently for all remaining orientations and then unioned with the primary solution.  This can only overcount the final cost, so it is enough to bound the sum of the repair costs orientation by orientation.

\begin{lemma}\label{lem:Ap_d_orient}
Fix a remaining orientation $\theta_r$ and a point $p$ on a line of orientation $\theta_r$.  Let $A_p$ be the event that $p$ survives the primary step and lies in a surviving block of LP mass at least $1$ along that line.  Then
$
\Prb[A_p] \le (k+1)e^{-k}.
$
\end{lemma}

\begin{proof}
Fix a line $\ell$ of orientation $\theta_r$ and order its points along $\ell$. In one
primary orientation, a point $i$ survives with probability $1-x_i$; hence after all
$k$ primary rounds it survives with probability $(1-x_i)^k$. Moreover, for any primary
orientation $\theta_t$ and any two distinct points on $\ell$, the two $\theta_t$-lines
through these points are distinct; otherwise a $\theta_t$-line would meet $\ell$ in two
points, forcing it to be $\ell$, contrary to $\theta_t\ne\theta_r$. Thus the collections
of primary-line shifts determining the survival of distinct points on $\ell$ are disjoint
and independent, so the survival events of distinct points on $\ell$ are independent.

Now repeat the weighted-case block analysis from the proof of~\Cref{lem:LpRp} with ``survives the primary step'' in place
of ``is unselected in Phase~I''. Let $D_p$ be the LP mass of the maximal surviving run
ending at $p$, and let $U_p$ be the LP mass of the maximal surviving run immediately after
$p$. If $A_p$ occurs, then $D_p+U_p\ge 1$.

For any $t\ge 0$, the same suffix argument as in the proof of~\Cref{lem:LpRp} gives
\[
\Prb[D_p\ge t]\le
\prod_{i\in J}(1-x_i)^k
\le
\exp\!\left(-k\sum_{i\in J}x_i\right)
\le e^{-kt},
\]
where $J$ is a minimal consecutive suffix ending at $p$ with LP mass at least $t$; if no such suffix
exists, the probability is $0$. Similarly, $\Prb[U_p\ge t]\le e^{-kt}$, and $D_p,U_p$
are independent. Thus $D_p$ and $U_p$ are stochastically dominated by independent
exponential random variables $E_1,E_2$ of rate $k$. Therefore,
\[
\Prb[A_p]\le \Prb[D_p+U_p\ge 1]
\le \Prb[E_1+E_2\ge 1]
= (k+1)e^{-k}.
\]
\end{proof}

Thus, we get that for each remaining orientation $r>k$,
\[
\E[\text{repair cost for orientation }\theta_r] \le \sum_{p \in P} x_p w(p) \Prb[A_p]
\le (k+1)e^{-k}\LP.
\]
Summing over the $d-k$ remaining orientations and adding the primary cost yields
\begin{equation}\label{eq:d_orient_master}
\E[\text{total cost}]
\le \bigl(k+(d-k)(k+1)e^{-k}\bigr)\LP.
\end{equation}

\dorient*

\begin{proof}
Take $k= \min(d, \lceil \ln d+\ln\ln d\rceil)$. For large enough $d$, $k$ equals $\lceil \ln d+\ln\ln d\rceil$, and thus $(d-k)(k+1)e^{-k}=O(1)$, yielding the claimed approximation ratio when substituted in~\eqref{eq:d_orient_master}.
\end{proof}

We can further generalize this result to the case when each object to be hit is a union of up to $h$ segments, each of which has one of $d$ orientations.

\begin{corollary}
\label{cor:union_of_segments}
  There exists a randomized algorithm that, given an instance of the weighted Hitting Set problem where each object is a union of at most $h$ segments of $d$ distinct orientations, computes a solution with an expected approximation factor of $h \cdot (\ln d+\ln\ln d+O(1))$ relative to the natural LP.
\end{corollary}

\begin{proof}
 Let $x$ be a feasible solution to the natural LP. Since each object $O$ is a union of at most $h$ segments, the LP constraint $\sum_{p \in O} x_p \ge 1$ implies there is some segment $s \subseteq O$ such that $\sum_{p \in s} x_p \ge 1/h$. Choose one such segment for each object, and define $y_p:=\min\{1,hx_p\}$. Then $0\le y_p\le 1$, $\sum_p w(p)y_p\le h\cdot\LP$, and for every chosen segment $s$ we have $\sum_{p\in s}y_p\ge 1$: if some $p\in s$ has $hx_p\ge 1$ this is immediate, and otherwise $\sum_{p\in s}y_p=h\sum_{p\in s}x_p\ge 1$. Thus $y$ is feasible for the chosen segment instance, and hitting these chosen segments hits every original object. Running the algorithm for $d$ orientations on $y$ yields the claimed expected approximation factor.
\end{proof}

\section{Local Search PTASes for $k$-restricted instances}
\label{sec:ptas}
In this section, we show that if we restrict the input so that each segment intersects
segments of the opposite orientation lying on at most $k$-distinct lines, then we
obtain a PTAS for the unweighted Hitting Set problem. 
The results of Raman and Ray~\cite{RR18} only work if the objects are non-piercing.
While Banik et al.~\cite{BanikICALP26} works for the case where each object has
a bounded amount of piercing, their results only work for the Set Cover and
Independent Set problems.
The restriction on the number
of such intersections is essential, as otherwise, as Section~\ref{sec:apxhardness}
shows, the problem is APX-hard.

\begin{definition}
A family $\cS=\cH\cup\cV$ of axis-parallel segments is \emph{$k$-restricted} if every
segment in $\cH$ intersects segments of $\cV$ lying on at most $k$ parallel lines, 
and every segment in $\cV$
intersects segments of $\cH$ lying on at most $k$ parallel lines.
\end{definition}

Observe that in a $k$-restricted instance, a segment can intersect
an unbounded number of segments of the opposite orientation. The only restriction
is that these segments lie on at most $k$ distinct parallel lines.

The local-search algorithm is standard. Fix an integer parameter $t$. Starting from any
feasible hitting set, repeatedly replace at most $t$ objects of the current solution by a
strictly smaller feasible set, as long as such an improving swap exists. For fixed $t$,
this process terminates in polynomial time. The key issue is to prove that a $t$-locally
optimal solution is globally near-optimal when $t=t(\varepsilon)$ is chosen large enough~\cite{RR18,DBLP:conf/walcom/AschnerKMY13,ChanH12,mustafa2010improved}.

The analysis relies on showing the existence of a suitable \emph{exchange graph} 
$G=(L\cup O, E)$ on a locally optimal solution $L$ and an optimal solution $O$. 
This exchange graph has two properties:
the \emph{local property} that for any $L'\subseteq L$,
$(L'\setminus L)\cup N_G(L')$ is also a hitting set for $\mathcal{S}$, where
$N_G(L')$ are the neighbors of $L'$ in the exchange graph $G$, and
the \emph{separator property} that $G$ comes from a hereditary family 
$\mathcal{G}$ of graphs with sub-linear size separators. 
That is, there is an absolute constant $\epsilon>0$ such that any graph 
$H\in\mathcal{G}$ has a separator of size $O(|V(H)|^{1-\epsilon})$\footnote{A balanced separator, 
or just a separator for a graph $G=(V,E)$ is a subset $S\subseteq V$ 
of vertices whose removal partitions the graph such that no resulting connected component contains more than $\alpha|V|$ vertices, where $\alpha < 1$ is a fixed constant, typically $1/2$ or $2/3$.
A graph admits a sub-linear separator if there exists an absolute constant $\epsilon>0$ such that $|S|=O(|V|^{1-\epsilon})$.
}.

Given two disjoint hitting sets $R$ and $B$ for a set $\mathcal{S}$ of axis-parallel segments, we show
the existence of a suitable exchange graph. A PTAS then follows by standard arguments.

\begin{lemma}
\label{lem:exchange}
Let $\cS=\cH\cup\cV$ be a set of $k$-restricted horizontal and vertical segments and let $R$ and $B$ be two disjoint hitting
sets for the segments $\cS$. Then, there is a local-exchange graph $G=(R\cup B,E)$ that satisfies
the local exchange and the separator properties.
\end{lemma}
\begin{proof}
The graph $G=(R\cup B, E)$ is constructed as follows. Starting with the empty graph on $R\cup B$,
for each segment $s\in\mathcal{S}$, let $R_s$ and $B_s$ denote respectively, the points of $R$ in $s$
and the points of $B$ in $s$. Suppose no point in $R_s$ is adjacent to a point in $B_s$, 
we choose an arbitrary pair of points, one in $R_s$ and the other in $B_s$ 
that are consecutive along $s$ and add the edge $\{r,b\}$
to $G$ as a straight-line segment along $s$. 
We claim that $G$ is the desired local-exchange graph. 
We say that the points $r$ and $s$ are \emph{responsible} for hitting $s$.

Let $R'\subseteq R$. We claim that $T=(R\setminus R')\cup N_G(R')$ is a hitting set. Indeed, for a
segment $s$, if the vertex $r$ responsible for hitting $s$ is not selected in $R'$, then $s$
is still hit by a point in the set $T$. If not, since the vertex $b\in B$ responsible for hitting
$s$ is adjacent to $r$, it implies that $s$ is hit by $b$. Therefore, $G$ satisfies the local-exchange
property.

Consider the natural embedding of $G$ in the plane, where the edges are drawn as straight-line segments
between the points. By construction, each edge in $G$ is only intersected by edges of the opposite orientation.
Since any edge lies along a segment, and the instance is $k$-restricted, it implies that each edge of
$G$ is crossed by at most $k$ other edges. 
That is $G$ is a \emph{$k$-planar graph}\footnote{A graph is $k$-planar if it has an embedding in the plane such that each edge is crossed by at most $k$ other edges}. $k$-planar graphs form a hereditary family as removing vertices
leaves the induced graph $k$-planar.

Since $k$-planar graphs have $O(\sqrt{k}n)$ edges~\cite{DBLP:journals/combinatorica/PachT97}, where $n=|R|+|B|$
and any two segments intersect in at most one point
we obtain a plane graph $H$ with $O(k^{3/2}n)$ vertices, by splitting each pair of crossing edges
by adding a new \emph{splitting vertex} at their intersection point. 
We assign a weight of 1 to each vertex of $G$ and a weight of $0$ to each splitting vertex. 
By the Planar Separator Theorem~\cite{LT79}, there is a balanced separator $S$ of $H$ of size $O(\sqrt{|V(H)|}$
such that both sides of the partition has weight at most $2/3$ the total weight of the vertices in $H$.
To obtain
a separator in $G$, for each vertex in $S$, if a splitting vertex $s$ lies in $S$, we add the four end-points
corresponding to the pair of crossing edges defined by $s$. The resulting set $S'$ is a 
balanced separator in $G$, and has size $O(|S|)$. Therefore, we obtain a sub-linear sized separator in $G$.
\end{proof}
\KRestricted*
\begin{proof}[Proof sketch]
Let $L$ be a $t$-locally optimal Hitting Set and let $O$ be an optimal Hitting Set. 
We can assume that $L\cap O=\emptyset$, as otherwise, we can remove the
points in $L\cap O$ and the segments hit by them, and consider the residual instance.
This can only improve the approximation factor. 
By Lemma~\ref{lem:exchange}, there is an exchange graph satisfying the local-exchange
property. Now, a PTAS follows directly via the techniques in~\cite{mustafa2010improved}.
\end{proof}

\section{APX-hardness}
\label{sec:apxhardness} 
The proof of the following theorem is based on the proof of APX-hardness of Minimum Hitting Set for orthogonal segments by Fekete et al.~\cite{FeketeEtAl2018}.

\Apxhard*
\begin{proof}
We give an approximation-preserving reduction from
\textsc{Max-2SAT(3)}, the maximum satisfiability problem for CNF formulae
with two literals per clause, in which each variable appears in at most three
clauses. This problem is APX-hard~\cite{hastad1999clique}.

Let $\Phi$ be an instance of \textsc{Max-2SAT(3)} with variables
$x_1,\ldots,x_n$ and clauses $C_1,\ldots,C_m$. We may assume that every
variable appears in at least one clause; isolated variables can be deleted.
Since each clause contains two literals, the total number of literal occurrences
is $2m$; therefore, $n \le 2m$.

Let $\OPT$ denote the maximum number of clauses of $\Phi$
that can be satisfied. Also, since a uniformly random truth assignment
satisfies each 2-literal-clause with probability at least $1/2$, we have
\[
        \OPT\ge \frac{m}{2}.
\]

We construct an instance of Minimum Hitting Set as follows. For each variable $x_i$, let $d_i\in\{1,2,3\}$ be the number of occurrences of $x_i$ in $\Phi$. We construct a variable gadget $G_i$ consisting of
$2d_i+4$ geometric objects: $2d_i+2$ horizontal segments and two vertical
lines. For each variable $x_i$, there is vertical line $L_i$ at the left end and there is a vertical line $R_i$ at the right end. These objects are arranged so that their intersection graph is an even cycle of length $2d_i+4$, and no point intersects three objects of the variable
gadget. Equivalently, the only points that hit two objects of $G_i$ are the
consecutive intersections along this cycle. Refer to Figure \ref{fig:max2sat-hs}.

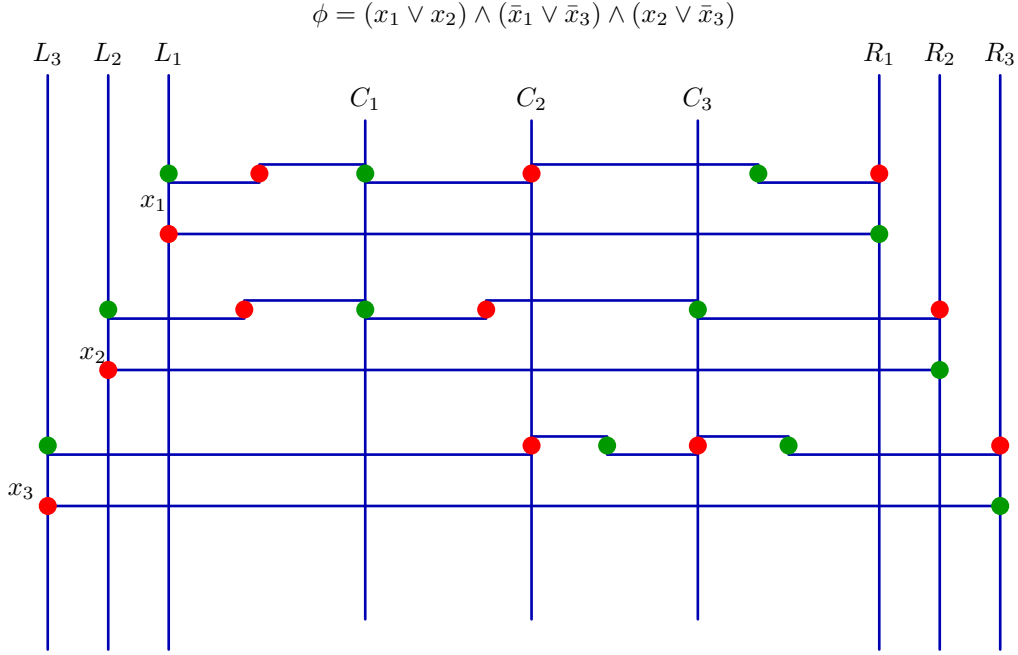
\begin{figure}[ht]
\centering
\begin{tikzpicture}[
    x=1cm,y=1cm,
    line cap=round,line join=round,
    seg/.style={blue!70!black, line width=1.0pt},
    lab/.style={font=\small},
    ptg/.style={circle, fill=green!60!black, inner sep=2.4pt},
    ptr/.style={circle, fill=red, inner sep=2.4pt}
]

\def\Lthree{0.0}
\def\Ltwo{0.8}
\def\Lone{1.6}

\def\Cone{4.2}
\def\Ctwo{6.4}
\def\Cthree{8.6}

\def\Rone{11.0}
\def\Rtwo{11.8}
\def\Rthree{12.6}

\def\Ytop{7.8}
\def\Ybot{0.2}

\def\xoneT{6.5}
\def\xoneB{5.7}

\def\xtwoT{4.7}
\def\xtwoB{3.9}

\def\xthreeT{2.9}
\def\xthreeB{2.1}

\def\dy{0.12}

\node[lab] at (6.3,8.6)
{$\phi=(x_1\vee x_2)\wedge(\bar{x}_1\vee \bar{x}_3)\wedge(x_2\vee \bar{x}_3)$};

\draw[seg] (\Lthree,\Ybot) -- (\Lthree,\Ytop);
\draw[seg] (\Ltwo,\Ybot) -- (\Ltwo,\Ytop);
\draw[seg] (\Lone,\Ybot) -- (\Lone,\Ytop);

\draw[seg] (\Cone,\Ybot+0.4) -- (\Cone,\Ytop-0.6);
\draw[seg] (\Ctwo,\Ybot+0.4) -- (\Ctwo,\Ytop-0.6);
\draw[seg] (\Cthree,\Ybot+0.4) -- (\Cthree,\Ytop-0.6);

\draw[seg] (\Rone,\Ybot) -- (\Rone,\Ytop);
\draw[seg] (\Rtwo,\Ybot) -- (\Rtwo,\Ytop);
\draw[seg] (\Rthree,\Ybot) -- (\Rthree,\Ytop);

\node[lab, above] at (\Lthree,\Ytop) {$L_3$};
\node[lab, above] at (\Ltwo,\Ytop) {$L_2$};
\node[lab, above] at (\Lone,\Ytop) {$L_1$};

\node[lab, above] at (\Cone,\Ytop-0.6) {$C_1$};
\node[lab, above] at (\Ctwo,\Ytop-0.6) {$C_2$};
\node[lab, above] at (\Cthree,\Ytop-0.6) {$C_3$};

\node[lab, above] at (\Rone,\Ytop) {$R_1$};
\node[lab, above] at (\Rtwo,\Ytop) {$R_2$};
\node[lab, above] at (\Rthree,\Ytop) {$R_3$};

\node[lab] at ({0.25*\Ltwo+0.75*\Lone},6.1) {$x_1$};
\node[lab] at ({0.25*\Lthree+0.75*\Ltwo},4.1) {$x_2$};
\node[lab] at (\Lthree-0.35,2.3) {$x_3$};

\draw[seg] (\Lone,\xoneT-\dy) -- (2.8,\xoneT-\dy);
\draw[seg] (2.8,\xoneT+\dy) -- (\Cone,\xoneT+\dy);
\draw[seg] (\Cone,\xoneT-\dy) -- (\Ctwo,\xoneT-\dy);
\draw[seg] (\Ctwo,\xoneT+\dy) -- (9.4,\xoneT+\dy);
\draw[seg] (9.4,\xoneT-\dy) -- (\Rone,\xoneT-\dy);

\draw[seg] (\Lone,\xoneB) -- (\Rone,\xoneB);

\node[ptg] at (\Lone,\xoneT) {};
\node[ptr] at (2.8,\xoneT) {};
\node[ptg] at (\Cone,\xoneT) {};
\node[ptr] at (\Ctwo,\xoneT) {};
\node[ptg] at (9.4,\xoneT) {};
\node[ptr] at (\Rone,\xoneT) {};

\node[ptr] at (\Lone,\xoneB) {};
\node[ptg] at (\Rone,\xoneB) {};

\draw[seg] (\Ltwo,\xtwoT-\dy) -- (2.6,\xtwoT-\dy);
\draw[seg] (2.6,\xtwoT+\dy) -- (\Cone,\xtwoT+\dy);
\draw[seg] (\Cone,\xtwoT-\dy) -- (5.8,\xtwoT-\dy);
\draw[seg] (5.8,\xtwoT+\dy) -- (\Cthree,\xtwoT+\dy);
\draw[seg] (\Cthree,\xtwoT-\dy) -- (\Rtwo,\xtwoT-\dy);

\draw[seg] (\Ltwo,\xtwoB) -- (\Rtwo,\xtwoB);

\node[ptg] at (\Ltwo,\xtwoT) {};
\node[ptr] at (2.6,\xtwoT) {};
\node[ptg] at (\Cone,\xtwoT) {};
\node[ptr] at (5.8,\xtwoT) {};
\node[ptg] at (\Cthree,\xtwoT) {};
\node[ptr] at (\Rtwo,\xtwoT) {};

\node[ptr] at (\Ltwo,\xtwoB) {};
\node[ptg] at (\Rtwo,\xtwoB) {};

\draw[seg] (\Lthree,\xthreeT-\dy) -- (\Ctwo,\xthreeT-\dy);
\draw[seg] (\Ctwo,\xthreeT+\dy) -- (7.4,\xthreeT+\dy);
\draw[seg] (7.4,\xthreeT-\dy) -- (\Cthree,\xthreeT-\dy);
\draw[seg] (\Cthree,\xthreeT+\dy) -- (9.8,\xthreeT+\dy);
\draw[seg] (9.8,\xthreeT-\dy) -- (\Rthree,\xthreeT-\dy);

\draw[seg] (\Lthree,\xthreeB) -- (\Rthree,\xthreeB);

\node[ptg] at (\Lthree,\xthreeT) {};
\node[ptr] at (\Ctwo,\xthreeT) {};
\node[ptg] at (7.4,\xthreeT) {};
\node[ptr] at (\Cthree,\xthreeT) {};
\node[ptg] at (9.8,\xthreeT) {};
\node[ptr] at (\Rthree,\xthreeT) {};

\node[ptr] at (\Lthree,\xthreeB) {};
\node[ptg] at (\Rthree,\xthreeB) {};

\end{tikzpicture}
\caption{A hitting-set instance corresponding to
$\phi=(x_1\vee x_2)\wedge(\bar{x}_1\vee \bar{x}_3)\wedge(x_2\vee \bar{x}_3)$.}
\label{fig:max2sat-hs}
\end{figure}

Color the consecutive intersection points of this cycle alternately red and
green. Since $G_i$ has $2d_i+4$ objects and every point hits at most two
objects of $G_i$, any hitting set for $G_i$ has size at least $b_i := d_{i}+2$.

Moreover, in a hitting set of size exactly $b_i$, every selected point must lie at an intersection of two consecutive objects of $G_i$. Therefore, it corresponds to a perfect matching of the cycle. Thus there are exactly two canonical optimum hitting sets for $G_i$:
the set of all red points and the set of all green points. We interpret the
green solution as setting $x_i=\mathrm{true}$ and the red solution as setting
$x_i=\mathrm{false}$.

For each clause $C_j$, introduce a vertical line $\ell_j$. If the literal
$x_i$ occurs in $C_j$, we make $\ell_j$ pass through the green point of
$G_i$ reserved for this occurrence. If the literal $\overline{x_i}$ occurs in
$C_j$, we make $\ell_j$ pass through the corresponding red point of $G_i$.
The lines and input points are placed such that these are the only points
at which a clause line passes through. Thus, each clause line passes through exactly two input points.
Let $ B := \sum_{i=1}^n b_i$.
Since $\sum_i d_i = 2m$, we get
$B= 6m$.

We first relate satisfying assignments to hitting sets. Given an optimal truth
assignment satisfying $\OPT$ clauses, choose in each variable gadget the canonical
red or green hitting set corresponding to the truth value of the variable.
This uses exactly $B$ points. Every satisfied clause line is hit at the
corresponding literal point. Each unsatisfied clause line can be hit by one
additional point placed anywhere on that line. Therefore, if $k$ denotes the
optimum solution size of the constructed hitting-set instance, then
$k \le B + (m-\OPT)$. In particular,
$k \le B+m \le 7m$.

Conversely, let $H$ be any hitting set for the constructed instance. We
normalize $H$ without increasing its cardinality. Consider a variable gadget
$G_i$. Since $G_i$ has $2b_i$ objects and no point hits more than two objects
of $G_i$, the points of $H$ used to hit objects of $G_i$ have cardinality at
least $b_i$.

If the selected points of $H$ for the variable gadget $G_i$ are consistent, that is, all are red or all are green, then we replace the points of $H$ used
inside $G_i$ by the corresponding canonical optimum set of $b_i$ red or
green points. This does not increase the number of points.

If the selected literal points are inconsistent, then $H$ uses literal points of both colors. Since the two canonical solutions are the only size-$b_i$ hitting
sets for $G_i$, the portion of $H$ used for $G_i$ has size at least $b_i+1$.
We replace it by one canonical optimum set of size $b_i$, choosing the
majority color among the selected literal points. Since $d_i\le 3$, at most one
clause line that was previously hit through a minority-colored literal point
becomes unhit. We use the one saved point to hit that clause line separately.
Again, the total number of points does not increase.

Applying this operation independently to all variable gadgets, we obtain a
hitting set of the form $H_v \cup H_c$,
with $|H_v| = B$ and $|H_v|+|H_c| \le |H|$.

The set $H_v$ is canonical in every variable gadget, and hence defines a truth
assignment. The set $H_c$ consists of additional points placed on clause lines.
Every clause not satisfied by the assignment induced by $H_v$ must be hit by
a point of $H_c$. Therefore, if $\operatorname{alg}$ is the number of clauses satisfied by the induced assignment, then
\[
        \operatorname{alg} \ge m-|H_c|.
\]

Now suppose there is a polynomial-time $(1+\varepsilon)$-approximation
algorithm for the hitting-set problem on horizontal segments and vertical
lines. Applied to the constructed instance, it returns a hitting set $H$ with
$|H| \le (1+\varepsilon)k$, and
after the normalization above,
$|H_c| \le |H|-B \le (1+\varepsilon)k-B$.
Hence,
\[
\begin{aligned}
        \operatorname{alg}
        &\ge m-|H_c| 
        &\ge m-(1+\varepsilon)k+B 
        &= (m+B-k)-\varepsilon k .
\end{aligned}
\]
Since $k\le B+(m-\OPT)$, we have
$m+B-k \ge \OPT$,
Thus, $\operatorname{alg}
        \ge \OPT-\varepsilon k$. 
Using $k\le 7m$ and $\OPT\ge m/2$, we obtain
\[
        \varepsilon k
        \le 7\varepsilon m
        \le 14\OPT.
\]
Therefore,
\[
        \operatorname{alg}
        \ge \left(1-14\varepsilon\right)\OPT.
\]
Consequently, a PTAS for minimum hitting set with horizontal segments and
vertical lines would imply a PTAS for \textsc{Max-2SAT(3)}, contradicting its
APX-hardness. Hence the hitting-set problem for horizontal segments and
vertical lines is APX-hard.
\end{proof}

\section{Conclusion}
\label{sec:conclusion}
We gave a short LP-rounding framework for hitting axis-parallel segments with weighted
points. The algorithm improves the immediate factor-$2$ bound by rounding one orientation
exactly and repairing the residual subproblem on the other orientation using exact
one-dimensional optimization. This yields a deterministic $(1+2/e)$-approximation for the
weighted problem and a deterministic $(1+1/(e-1))$-approximation in the unweighted case.
For bounded-complexity subclasses, local search yields a PTAS for the Hitting Set problem, 
and we also extend our results for segments in $d$ orientations.

Two natural directions remain open. First, it would be interesting to determine whether
the analysis can be strengthened further, or whether these bounds are close to the true
integrality behavior of the natural LP. For the Set Cover and Independent Set problems,
the best known approximation is the straightforward factor-$2$ bound.
The integrality gap of the natural LP-formulation for 
Independent Set was shown to be $2$~\cite{DBLP:journals/jocg/CaoduroCPW23}. Hence, breaking the factor-2
for the independent set problem cannot go via rounding the natural LP and new ideas are required.
Finally, the approach suggests exploring whether
similar one-orientation-first rounding schemes can be extended to richer families such as thin rectangles.
\bibliographystyle{plain}
\bibliography{ref}
\newpage

\appendix

\section{Deferred Proof for Stochastic Dominance}
\label{app:stochastic-dominance}
\StochasticSumDominance*
\begin{proof}
Recall that $X\st Y$ means
\[
\Prb[X\ge s]\le \Prb[Y\ge s]
\qquad\text{for every } s\in \mathbb{R}.
\]
We use the standard equivalent characterization of first-order stochastic dominance: if
$A\st B$ and $f$ is a bounded nondecreasing function, then
$\Ex[f(A)]\le \Ex[f(B)]$.
Fix $t\in\mathbb{R}$ and define
\[
g(z):=\Prb[Y_1\ge t-z].
\]
The function $g$ is nondecreasing in $z$. Using independence of $X_1$ and
$X_2$, then $X_1\st Y_1$, then $X_2\st Y_2$, and finally independence of
$Y_1$ and $Y_2$, we get
\[
\begin{aligned}
\Prb[X_1+X_2\ge t]
&= \Ex\!\left[\Prb[X_1\ge t-X_2]\right]\\
&\le \Ex\!\left[\Prb[Y_1\ge t-X_2]\right]\\
&= \Ex[g(X_2)]\\
&\le \Ex[g(Y_2)]\\
&= \Ex\!\left[\Prb[Y_1\ge t-Y_2]\right]\\
&= \Prb[Y_1+Y_2\ge t].
\end{aligned}
\]
Therefore
\[
X_1+X_2 \st Y_1+Y_2.
\qedhere
\]
\end{proof}

\section{Tightness of the Unweighted Analysis}
\label{app:unweighted-tightness}

We record a family of instances showing that the analysis in \Cref{sec:unweighted} is
asymptotically tight. Let $k$ tend to infinity, and set $n:=k^2$. The point set is the
$n\times n$ grid $P=[n]\times[n]$. The input contains every horizontal and vertical
segment consisting of exactly $k$ consecutive grid points.

First, $\OPT=n^2/k=k^3$. Indeed, the $k$ disjoint length-$k$ horizontal segments in each
row show that every feasible solution has size at least $nk=k^3$. Conversely, the set
\[
Q^\star:=\{(i,j)\in[n]\times[n] : j-i\equiv 0 \pmod{k}\}
\]
has exactly $k$ points in every row and every column, with consecutive chosen points spaced
exactly $k$ apart. Hence every length-$k$ horizontal or vertical segment contains one
point of $Q^\star$, so $|Q^\star|=k^3=n^2/k$.

The same disjoint-segment lower bound also applies to the LP, and the uniform assignment
$x_p=1/k$ for every grid point is feasible. Hence this assignment is an optimal LP
solution. We now run the algorithm with this symmetric optimum. On each row, Phase~I
chooses exactly one residue class modulo $k$, independently and uniformly from row to row.
Thus Phase~I selects exactly $k$ points on each row, and
\[
\cost(\text{Phase~I})=k^3=\OPT.
\]

It remains to compute the expected repair cost on the columns. Fix a column. A row is
already selected in this column with probability $p:=1/k$, independently across rows. Let
$q:=1-p$ and let $Z$ be the Phase~II repair cost on this column. If a maximal run of
unselected points has length $L$, then the residual vertical segments inside this run are
exactly the length-$k$ intervals contained in the run, and the optimum one-dimensional
repair cost is $\lfloor L/k\rfloor$: the lower bound comes from the disjoint length-$k$
intervals in the run, and selecting every $k$th point attains it.

For $a\ge 1$, let $N_a$ be the number of unselected runs in the column whose length is at
least $ak$. Then $Z=\sum_{a=1}^{k}N_a$. A run of length at least $ak$ starts in the first
row with probability $q^{ak}$; it starts at any later row with probability $pq^{ak}$,
requiring the previous row to be selected and the next $ak$ rows to be unselected.
Therefore, with $\alpha:=q^k=(1-1/k)^k$,
\[
\Ex[Z]
=
\sum_{a=1}^{k}\left(q^{ak}+(n-ak)pq^{ak}\right)
=
\sum_{a=1}^{k}(k-a+1)\alpha^a \geq k \frac{\alpha (1 - \alpha^k)}{1- \alpha} - \frac{\alpha}{(1-\alpha)^2},
\]
using $\sum_{a=1}^{k} (a-1) \alpha^a \le \sum_{a=1}^{\infty} a \cdot \alpha^a = \alpha/(1-\alpha)^2$. Now, there are $k^2$ columns each contributing $\Ex[Z]$ to the expected cost of Phase~II and $\OPT = k^3$. Hence, the ratio of the expected cost of Phase~II and $\OPT$ satisfies:

\begin{equation}\label{eq:tight}
      \frac{\alpha(1 - \alpha^k)}{1 - \alpha} - \frac{1}{k}\frac{\alpha}{(1 - \alpha)^2} \leq \frac{\Ex[\cost(\text{Phase II})]}{\OPT} \leq \frac{1}{e-1},
\end{equation}
 where the second inequality follows from our analysis of the unweighted case in~\Cref{sec:unweighted}. Now, as $k\to\infty$, we have $\alpha=(1-1/k)^k\to e^{-1}$, thus the left hand side of~\eqref{eq:tight} tends to $e^{-1}/(1-e^{-1}) = 1/(e-1)$. Thus, the approximation factor $(1 + 1/(e-1))$ obtained in the unweighted analysis is asymptotically tight for the algorithm.

\section{Derandomization}
\label{sec:derandomization}
In our algorithm for axis-parallel line segments, only Phase~1 is randomized. Each horizontal line $h$ receives a shift
$U_h\in[0,1)$, and Phase~2 is deterministic once all shifts are fixed. Let the horizontal
lines containing points of $P$ be $h_1,h_2,\dots,h_r$. For a shift vector
$u=(u_1,\dots,u_r)\in[0,1)^r$, let $C(u)$ denote the final solution cost produced by the
algorithm when line $h_i$ uses shift $u_i$.

The randomized algorithm samples $U=(U_1,\dots,U_r)$ with independent coordinates
$U_i\sim \mathrm{Unif}[0,1)$ and returns cost $C(U)$. We have that $\Ex[C(U)] \le \alpha\,\OPT,$
where $\alpha=1+2/e$ in the weighted case and $\alpha=1+1/(e-1)$ in the unweighted case. Define conditional expectations
\[
\Phi_k(u_1,\dots,u_k)
:=
\Ex\!\left[C(U)\,\middle|\,U_1=u_1,\dots,U_k=u_k\right]
\qquad (k=0,1,\dots,r),
\]
with $\Phi_0=\Ex[C(U)]$. Suppose shifts $u_1,\dots,u_{k-1}$ are already fixed. Consider
\[
g_k(t):=\Phi_k(u_1,\dots,u_{k-1},t)
\qquad\text{for } t\in[0,1).
\]
By the tower property,
\[
\Ex_{U_k}[g_k(U_k)] = \Phi_{k-1}(u_1,\dots,u_{k-1}).
\]
Therefore some $t^\star\in[0,1)$ satisfies
\[
g_k(t^\star)\le \Phi_{k-1}(u_1,\dots,u_{k-1}).
\]
Set $u_k:=t^\star$. Repeating this for $k=1,\dots,r$ yields a deterministic shift vector
$u$ with
\[
C(u)=\Phi_r(u_1,\dots,u_r)\le \Phi_0=\Ex[C(U)].
\]

To make the search finite, fix a line $h_k$ with points $p_1,\dots,p_m$ in increasing
$x$-order and partial sums $a_j=\sum_{i=1}^{j}x_{p_i}$. The selection pattern on $h_k$
changes only when the shift crosses one of the fractional parts $\{a_j\}$, so $[0,1)$ is
partitioned into at most $m+1$ intervals on which the Phase~1 outcome is constant.
Therefore, for each step of the conditional-expectation procedure, it suffices to test one
representative from each interval and choose the best.

For any such representative $t$, the value $g_k(t)$ can be computed in polynomial time.
The expected Phase~1 cost is immediate from the fixed shifts and the marginals $x_p$ on
the remaining lines. For Phase~2, enumerate all possible blocks, i.e., all consecutive
runs of points on vertical lines. For each block $B$, compute the optimum
one-dimensional (weighted) hitting set for the segments contained in $B$, and multiply this
value by the probability that $B$ is present as a residual block. This probability is a
product of the relevant selection and non-selection probabilities on the points of $B$ and
its boundary points; fixed lines contribute probabilities $0$ or $1$, while unfixed lines
contribute the marginals $x_p$ or $1-x_p$. Since there are only polynomially many blocks
and each block is a one-dimensional instance, this gives a polynomial-time evaluation of
$g_k(t)$.

\begin{theorem}
There is a deterministic implementation of the algorithm whose output cost is at most the
expected cost of the randomized algorithm. In particular, it achieves approximation ratio
$1+2/e$ in the weighted case and $1+1/(e-1)$ in the unweighted case.
\end{theorem}

The same conditional-expectation procedure derandomizes the $d$-orientation algorithm of
\Cref{sec:orientation}: fix the random shifts on the primary-orientation lines one at a
time and evaluate each conditional expectation by the same polynomial-time enumeration of
residual blocks on the remaining orientations.

\section{Discrete $3$-Slope Line Cover is APX-hard}\label{app:3slc}
\textsc{Discrete 3-Slope-Line-Cover} ($\mathsf{3SLC}$) refers to the problem of computing a minimum-cardinality hitting set for a set $\mathcal{S}$ of lines having three distinct slopes with a subset of a given set $P$ of points. In this section, we prove that this problem is APX-hard \cite{FeketeEtAl2018}.
 
\begin{definition}(\textsf{$L$-reduction}) \cite{DBLP:journals/comgeo/ChanG14}
A pair of functions $f$ and $g$, both computable in polynomial time, is an \textsf{$L$-reduction} from an optimization problem $\mathcal{P}$ to an optimization problem $\mathcal{P'}$ if there are positive constants $\alpha$ and $\beta$ such that for each instance $x$ of $\mathcal{P}$, the following hold:
\begin{itemize}
    \item (L1) The function $f$ maps instances of $\mathcal{P}$ to instances of $\mathcal{P'}$ such that OPT($f(x)$) $\leq \alpha \cdot OPT(x)$.
    \item (L2) The function $g$ maps feasible solutions $y$ of $f(x)$ to feasible solutions $g(y)$ of $x$ such that $|c_{x}(g(y)) - OPT(x)| \leq \beta . |c_{f(x)}(y) - OPT(f(x))|$, where $c_x$ and $c_{f(x)}$ are the cost functions of the instances $x$ and $f(x)$ respectively.
\end{itemize}
\end{definition}

\threeline*
\begin{proof}
We give an $L$-reduction from \textsc{Special-3SC}, which is APX-hard
\cite{DBLP:journals/comgeo/ChanG14}. Recall that an instance of \textsc{Special-3SC} is obtained
from a cubic graph $G=(V,E)$. For every vertex $v_t\in V$, there are four
elements $w_t,x_t,y_t,z_t$, and for every edge $e_i\in E$, there is an element
$a_i$. If the three edges incident to $v_t$ are $e_i,e_j,e_k$, with
$i<j<k$, then the instance contains the five sets $\{w_t,a_i\}, \{w_t,x_t\}, \{x_t,y_t,a_j\}, \{y_t,z_t\}, \{z_t,a_k\}$.
Moreover, every element $a_i$ belongs to exactly two sets.

We first construct an equivalent finite discrete $\mathsf{3SLC}$ instance in the dual plane.
Let \(L_1,L_2,L_3\) be three vertical lines, from left to right. We place the
elements \(w_t,y_t\) on \(L_1\), the elements \(x_t,z_t\) on \(L_2\), and the
elements \(a_i\) on \(L_3\). The placement is chosen so that, for each
\(v_t\) with incident edges \(e_i,e_j,e_k\), \(i<j<k\), there are five
non-vertical candidate lines whose intersections with the constructed point set
are exactly
\[
        \{w_t,a_i\},\quad
        \{w_t,x_t\},\quad
        \{x_t,y_t,a_j\},\quad
        \{y_t,z_t\},\quad
        \{z_t,a_k\}.
\]
We briefly justify that such a placement can be found in polynomial time. Put
the edge points \(a_1,\ldots,a_n\) on \(L_3\) in this order. For each vertex
\(v_t\), reserve two small open intervals \(I_t^1\subset L_1\) and
\(I_t^2\subset L_2\), with these intervals pairwise disjoint over different
vertices, such that every line joining a point of \(I_t^1\) to an edge point
\(a_r\) meets \(L_2\) inside \(I_t^2\). Such intervals are obtained, for
example, by taking disjoint thin vertical cross-sections of triangles with
common base on \(L_3\).

Now choose \(y_t\in I_t^1\) and let \(x_t\) be the intersection of
\(\overline{y_ta_j}\) with \(L_2\). Choose \(w_t\in I_t^1\) and
\(z_t\in I_t^2\). At each step, the forbidden choices are only those causing
one of the five intended candidate lines to pass through an unintended point.
There are finitely many such forbidden positions. 
Hence we can choose rational
points avoiding all of them.
See Figure 
\ref{fig:3sc_3slc_reduction}.
\medskip
\begin{figure}[htpb]
  \begin{center}
    \resizebox{145mm}{!} {\includegraphics *{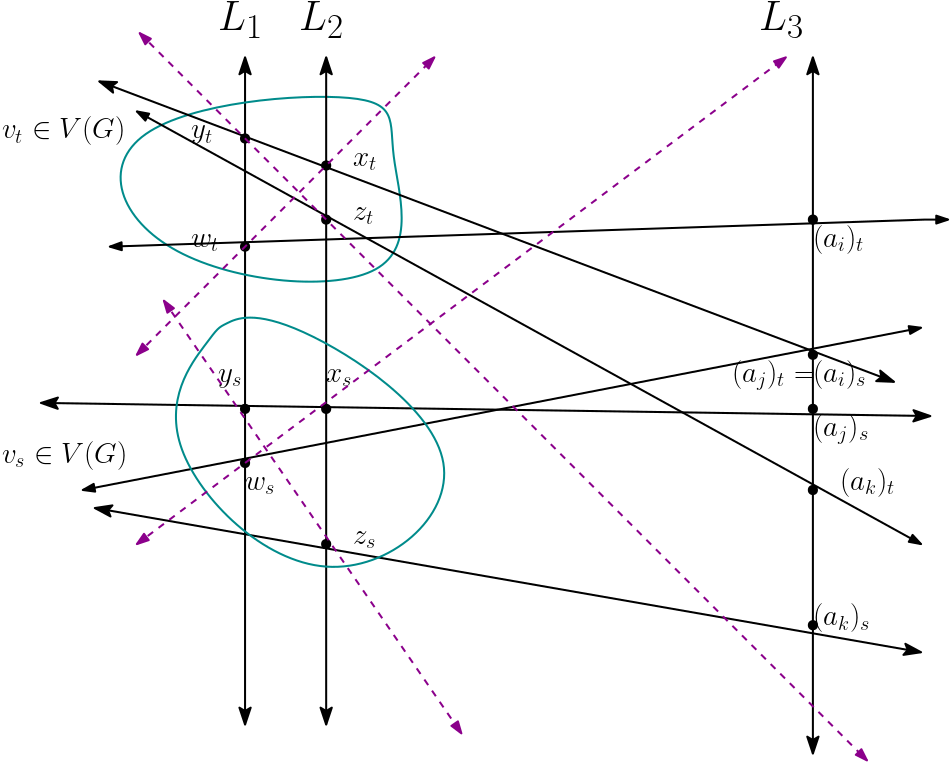}}
    \caption {
Two vertex gadgets for $v_s, v_t \in V$ in the dual construction, where $(v_s, v_t) \in E$. For each vertex \(v_t\) incident
to edges \(e_i,e_j,e_k\), \(i<j<k\), the points \(w_t,y_t\) lie on \(L_1\) and
\(x_t,z_t\) lie on \(L_2\). The three solid lines encode the sets
\(\{w_t,a_i\}\), \(\{x_t,y_t,a_j\}\), and \(\{z_t,a_k\}\); the dashed magenta
lines encode the non-edge sets \(\{w_t,x_t\}\) and \(\{y_t,z_t\}\).
}
  \label{fig:3sc_3slc_reduction}
  \end{center}
\end{figure}

It follows that the constructed dual finite line-cover instance has exactly the same set system as the \textsc{Special-3SC} instance: selecting a candidate line is
equivalent to selecting the corresponding set, and a subfamily covers all points
if and only if the corresponding sets cover the universe. Therefore
\[
        \operatorname{OPT}_{\mathrm{dual}}=\operatorname{OPT}_{\mathrm{SC}},
\]
where $\operatorname{OPT}_{\mathrm{dual}}$ denotes the optimal solution value for the dual line-cover instance and $\operatorname{OPT}_{\mathrm{SC}}$ denotes the optimal solution value for the \textsc{Special-3SC} instance.

It remains to pass from this dual finite line-cover instance to
\textsc{Discrete 3-Slope-Line-Cover}. We use the standard point-line duality
\[
        p=(p_x,p_y) \longleftrightarrow p^* : b=p_xa-p_y,
\]
and
\[
        \ell : y=ax-b \longleftrightarrow \ell^*=(a,b),
\]
which preserves incidences. The three vertical lines \(L_1,L_2,L_3\) in the
dual plane become three slopes in the primal plane. The dual points become the
input lines \(S\), and the dual candidate lines become the finite candidate
point set \(P\). Hence we obtain an instance of \textsc{Discrete 3-Slope Line
Cover}. Since incidence and cardinality are preserved, feasible covers in the
dual line-cover instance are in one-to-one correspondence with feasible hitting
sets in the resulting \textsc{Discrete 3-Slope-Line-Cover} instance.

Thus, for the reduction \(f\),
\[
        \operatorname{OPT}_{3\mathrm{SLC}}(f(I))
        = \operatorname{OPT}_{\mathrm{SC}}(I),
\]
so condition (L1) of an \(L\)-reduction holds with \(\alpha=1\). Given any
feasible solution to \(f(I)\), mapping each selected point back by duality gives
the corresponding selected set of the \textsc{Special-3SC} instance, with the
same cardinality. Therefore
\[
\left|
c_{\mathrm{SC}}(g(Y))-\operatorname{OPT}_{\mathrm{SC}}(I)
\right|
=
\left|
c_{3\mathrm{SLC}}(Y)-\operatorname{OPT}_{3\mathrm{SLC}}(f(I))
\right|,
\]
so condition (L2) holds with \(\beta=1\). The construction is polynomial-time
computable. Hence \textsc{Discrete 3-Slope-Line-Cover} is APX-hard.
\end{proof}

\end{document}